\documentclass[letterpaper, 10 pt, conference]{ieeeconf}
\IEEEoverridecommandlockouts
\usepackage{placeins}
\usepackage{cite}
\usepackage{amsmath,amssymb,amsfonts}
\usepackage{algorithm}
\usepackage{algpseudocode}
\usepackage{graphicx}
\usepackage{subcaption}
\usepackage{textcomp}
\usepackage{xcolor}
\usepackage{hyperref}
\usepackage{float}
\usepackage{makecell}
\usepackage{mathtools}
\usepackage{tikz}

\newtheorem{theorem}{Theorem}[section]
\newtheorem{remark}{Remark}

\newtheorem{assumption}{Assumption}
\newtheorem{proposition}{Proposition}
\newtheorem{lemma}{Lemma}

\newtheorem{problem}{Problem}
\newcommand{\QEDn}{\hfill $\blacksquare$}

\allowdisplaybreaks

\usepackage[normalem]{ulem}
\title{\LARGE \bf 
Finite-Sample Conformal Coverage Recovery via Fusion under Degraded Local Guarantees in Occupancy Map Estimation}
\author{ Ritvik Mahajan, Aneesh Raghavan, Karl Henrik Johansson
\thanks{This work was partly supported by the Knut and Alice Wallenberg
Foundation (Wallenberg Scholar grant), the Swedish Research Council (Distinguished Professor grant 2017-01078), and the Swedish Foundation for Strategic
Research (SUCCESS FUS21-0026).}%
\thanks{The authors are with the Department of Decision and Control Systems, School of Electrical Engineering and Computer Science, KTH Royal Institute of Technology,  and also with Digital Futures, Stockholm, Sweden.{\tt\small \{ritvikm,aneesh,kallej\}@kth.se}}%
}

\begin{document}
\maketitle
\thispagestyle{empty}
\pagestyle{empty}

\begin{abstract}
Accurate and reliable environmental mapping is a fundamental requirement for multi-robot autonomy. While continuous mapping techniques like Gaussian Process Occupancy Mapping (GPOM) provide rich spatial correlation and uncertainty estimates, they lack formal, finite-sample guarantees on their predictive reliability. Conformal prediction can equip each robot's local map with a distribution-free coverage guarantee, but this local guarantee degrades in practice: temporal correlation along a robot's trajectory breaks the exchangeability on which conformal calibration relies, and each robot observes only a spatially limited, non-uniform portion of the environment. Taking these degraded per-agent guarantees as given, we develop a distributed fusion algorithm that recovers the desired coverage across the team. Robots exchange only lightweight scalar e-values with their neighbors, and a receiver fuses them using a per-neighborhood miscoverage budget and an uncertainty-attenuated fusion operator. We prove that the fused set-valued map recovers the target user-specified coverage level regardless of the communication graph topology or the underlying sensor noise distribution. However, a drawback is that wherever the fused evidence is insufficient, the map declines to commit and returns both labels (free and occupied), leaving a significant fraction of the domain unclassified rather than thresholded into a single decision. Simulated multi-agent mapping experiments demonstrate that the fused predictor reliably meets its theoretical coverage bounds, and illustrate that denser communication topologies significantly enhance map efficiency by shrinking this unclassified fraction.
\end{abstract}

\section{Introduction}\label{sec:introduction}

Having an accurate map of the environment is fundamental to mobile robot autonomy. A likelihood map assigns each location a continuous occupancy score, a probability or log-odds value, whereas the occupancy map is the binary free/occupied classification obtained by applying a thresholding policy to that likelihood map. Mapping algorithms overwhelmingly estimate the likelihood map, leaving the occupancy map, the object actually consumed by planning and safety layers, to an unguaranteed thresholding step.
The classical Occupancy Grid Map (OGM) framework discretizes the workspace and treats cell occupancies independently \cite{elfes1989, thrun2005probabilistic}. 
While computationally tractable, this assumption discards spatial correlations and fixes map resolution.
Continuous alternatives, such as Gaussian Process Occupancy Mapping (GPOM) \cite{ocallaghan2012gpom}, overcome these limitations by providing spatial correlation and uncertainty estimates. 
Recent scalable extensions include incremental updates \cite{jadidi2018gpom} and Hilbert Maps \cite{ramos2016hilbert, senanayake2017bhm, senanayake2018continuous}. 
However, their uncertainty estimates rely heavily on heuristics or exact hyperparameter tuning, lacking formal finite-sample guarantees on map reliability.

Tasks such as exploration and search-and-rescue are better served by multi-robot teams covering larger areas efficiently, which requires systematic map-merging methods. 
Existing continuous fusion approaches handle intermittent connectivity and communication constraints using Gaussian Mixture Models \cite{dong2022mrgmmapping}, distributed sparse GPs \cite{di2022distributed}, or distributed factor graphs \cite{mcgann2024imesa}. 
Despite this progress, existing methods fail to provide formal, distribution-free guarantees on the fraction of map queries answered correctly. 
This is a critical gap for safety-critical deployment because heterogeneous exploration naturally results in non-uniform data quality across the environment.

Conformal Prediction (CP) is a distribution-free framework for constructing prediction sets with provable finite-sample coverage guarantees \cite{vovk2005alrw, papadopoulos2002icm}. 
Recent advances have extended CP to handle online data streams \cite{gibbs2021aci}, label noise \cite{einbinder2024label}, safe control planning \cite{lindemann2024survey}, and distributed computation via message passing \cite{wen2025distributed}. Specifically in mapping and navigation, recent work has begun leveraging CP to quantify uncertainty in 3D semantic occupancy in autonomous driving \cite{su2025alphaocc}, for semantic navigation in unknown environments under perceptual uncertainty \cite{sundarsingh2026safe}, and to extract probabilistic safety bounds for UGV trajectory planning \cite{wei2024ugv}.
These properties make CP a natural tool for certifying map predictions. 
However, providing formal finite-sample guarantees for distributed, multi-agent occupancy fusion under correlated, non-i.i.d.\ data remains an open challenge, which this work addresses. The proposed approach is agnostic to the choice of likelihood-map estimator, and instead targets the occupancy map it induces: each agent equips its own thresholded map with a conformal coverage guarantee, but since spatially and temporally correlated data degrade this per-agent guarantee below the target level, we pose the problem as retrieving the desired coverage by fusing the degraded local guarantees across the team, without exchanging raw data.

\subsection{Motivating Experiments}\label{subsec:motivation}

To make the challenge concrete, we run single- and multi-agent studies in the simulation environment described in Section~\ref{sec:simulation} (with full setup and noise levels there). As external baselines, we map using OGM (grid-based) and IGPOM (GP-based) with a single robot traversing the entire workspace ($989$ scans) and report ROC-AUC against the known ground truth. Two limitations emerge. First, sensing noise degrades the grid-based OGM (ROC-AUC $0.86\!\to\!0.82$ as the LiDAR range-noise standard deviation $\sigma_r$ grows from $0$ to $0.1$\,m, Fig.~\ref{fig:mot_noise}); the variance-aware IGPOM resists noise but is far costlier ($\sim\!140$\,s vs.\ $\sim\!8$\,s per map). Second, collecting more data does not close this gap: as the retained-scan fraction grows from $0.1$ to $1$, IGPOM's ROC-AUC stays within $[0.87,0.93]$ and never approaches $1$, while its GP training time rises roughly tenfold ($14\!\to\!134$\,s, Fig.~\ref{fig:mot_step}). The achievable single-agent quality is thus bounded away from perfect and cannot be bought with more measurements.

Collaboration, in contrast, does help because a single robot observes only part of the workspace. Fig.~\ref{fig:mot_collab} reports ROC-AUC over the whole workspace as more agents' local GPs are fused (each agent's bare-Mat\'ern GP, the pre-fusion local estimator of Section~\ref{subsec:local_est}, combined by a gate-masked probability average). A lone agent, informative only on its own ${\sim}60\%$ of the map, attains $0.80$; two complementary agents already reach $0.93$ (e.g.\ robots~4 and~5 individually score $0.79$ each but $0.92$ fused), and the five-agent team reaches $0.95$.

Yet all of these methods share a deeper limitation: each outputs a probabilistic map that must be thresholded heuristically to obtain a classified (occupied/free) map, and that threshold carries no guarantee on the fraction of misclassified queries. Fig.~\ref{fig:mot_single} shows the diffuse occupancy-probability field of a single agent's pre-fusion local estimator, the bare-Mat\'{e}rn GP of Section~\ref{subsec:local_est}, sharpening around observed structure and degrading under sparse, noisy data. Choosing the classification threshold is itself ill-posed: Fig.~\ref{fig:mot_f1} sweeps the threshold $\tau$ on the occupied-class F1 of one such map (full data, $\sigma_r=0.1$\,m), whose peak is $0.38$ at $\tau\approx0.6$ and collapses on either side; and because the occupied class is rare (${\sim}5\%$), accuracy is uninformative (a trivial all-free map scores ${\sim}95\%$). No principled, condition-independent threshold exists, and none of these maps bounds its misclassification rate. This motivates our approach: rather than thresholding a probability field into a single label per cell, we fuse such agents into a set-valued map that, at each query, returns the set of occupancy labels it deems admissible. This set carries a finite-sample guarantee at a user-chosen level $\alpha$: across queries, the true occupancy label is excluded from the returned set at most a fraction $\alpha$ of the time (equivalently, the fused map covers the true label with probability at least $1-\alpha$), regardless of the sensor noise or the amount of data.
\begin{figure*}[t]
  \centering
  \begin{subfigure}{0.32\textwidth}
    \includegraphics[width=\linewidth]{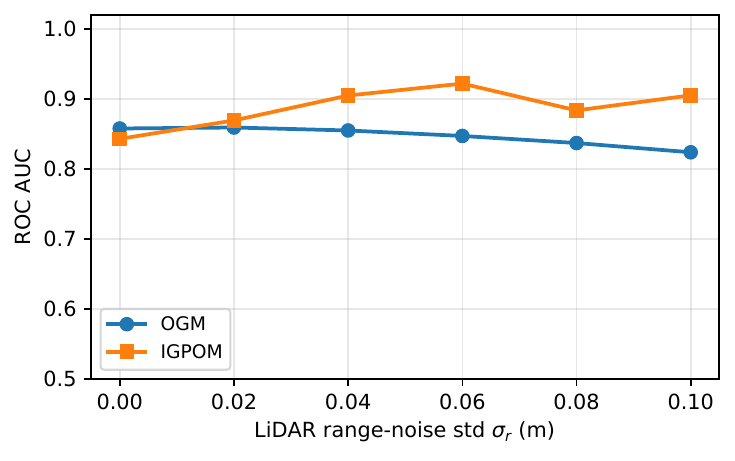}
    \caption{noise degrades quality}\label{fig:mot_noise}
  \end{subfigure}\hfill
  \begin{subfigure}{0.32\textwidth}
    \includegraphics[width=\linewidth]{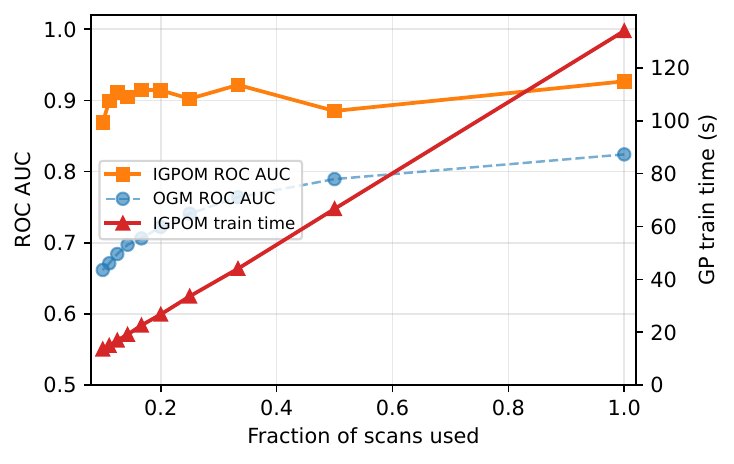}
    \caption{more data leads to more compute}\label{fig:mot_step}
  \end{subfigure}\hfill
  \begin{subfigure}{0.32\textwidth}
    \includegraphics[width=\linewidth]{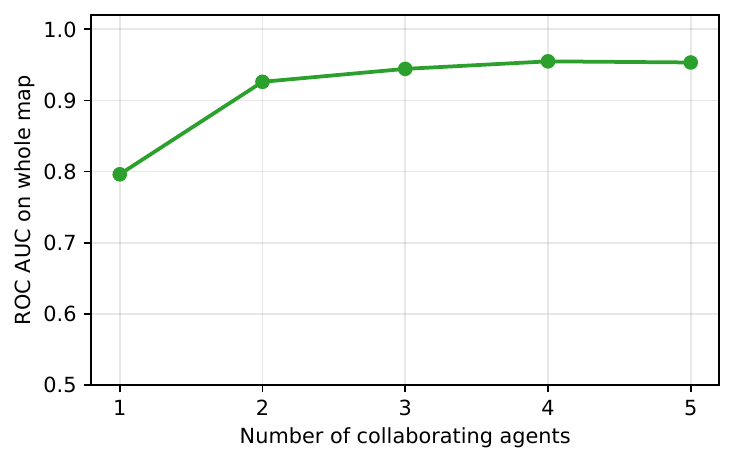}
    \caption{collaboration extends coverage}\label{fig:mot_collab}
  \end{subfigure}
  \caption{Motivating experiments (ROC-AUC vs.\ ground truth). \textbf{(a)} A single robot over the whole workspace: the grid-based OGM degrades as the LiDAR range-noise $\sigma_r$ grows, while the variance-aware IGPOM resists noise but costs ${\sim}18\times$ more and saturates. \textbf{(b)} The same robot, varying the retained-scan fraction: IGPOM's ROC-AUC stays flat while GP training time rises ${\sim}10\times$. \textbf{(c)} Multi-agent collaboration: fusing more agents' local GPs raises the whole-workspace ROC-AUC from $0.80$ (one agent) to $0.95$ (five agents), since collaboration maps regions that a single robot never observes.}
  \label{fig:mot_quant}
\end{figure*}

\begin{figure}[t]
  \centering
  \begin{subfigure}[t]{0.32\columnwidth}
    \includegraphics[width=\linewidth]{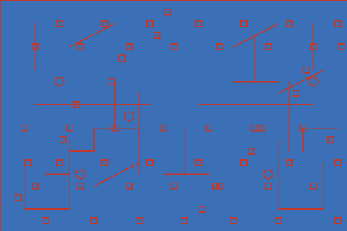}
    \caption{ground truth}\label{fig:mot_single_gt}
  \end{subfigure}\hfill
  \begin{subfigure}[t]{0.32\columnwidth}
    \includegraphics[width=\linewidth]{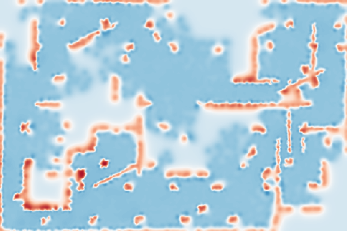}
    \caption{full data, low noise}\label{fig:mot_single_nom}
  \end{subfigure}\hfill
  \begin{subfigure}[t]{0.32\columnwidth}
    \includegraphics[width=\linewidth]{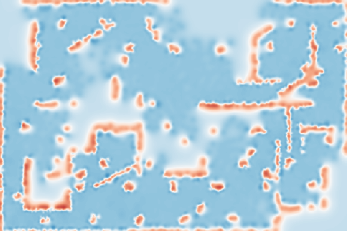}
    \caption{sparse + high noise}\label{fig:mot_single_deg}
  \end{subfigure}
  \caption{Occupancy probability of a single agent's pre-fusion local estimator over the whole workspace, beside the ground truth. In (b),(c), the color encodes $P(\mathrm{occupied})$ from $0$ (blue) to $1$ (red); in (a), blue/red mark free/occupied ground truth. The map is a diffuse probability field, sharper around observed structure and degrading under sparse/noisy data.}
  \label{fig:mot_single}
\end{figure}

\begin{figure}[t]
  \centering
  \includegraphics[width=\columnwidth]{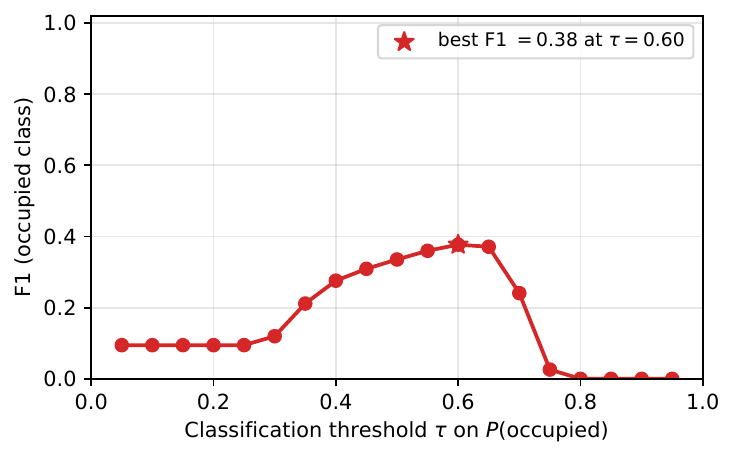}
  \caption{Occupied-class F1 versus classification threshold $\tau$ for a single-agent map (full data, $\sigma_r=0.1$\,m). The best F1 is $0.38$, at $\tau\approx0.6$, and collapses on either side.}
  \label{fig:mot_f1}
  \vspace{-0.7cm}
\end{figure}

\subsection{Problem Description}
Multiple agents are gathering data. Using the collected data in a static environment, the objective is to estimate an occupancy map without exchanging data. There are various algorithms for estimating the likelihood map, such as OGM, GPOM, and Hilbert maps. However, no algorithm formally guarantees the accuracy of the occupancy map obtained by applying a given thresholding policy to the likelihood map. Given the data and any choice of likelihood estimation algorithm, we assume that there exist thresholds that achieve coverage guarantees, even when the data are spatially and temporally correlated. We consider the problem of recovering coverage guarantees from degraded local guarantees. Formally, we cast distributed occupancy mapping as a coverage-retrieval problem (Problem~\ref{prob:main}): each agent applies split conformal prediction to its local likelihood map and returns a set-valued occupancy prediction whose miscoverage can exceed the target level $\alpha$, since its trajectory data are temporally correlated and cover only part of the workspace, and the task is to design a per-agent fusion rule over the communication neighborhood that restores level $\alpha$ while agents share only scalar e-values on a common query set and no raw measurements.

\subsection{Contributions}\label{subsec:contributions}
To solve the coverage-retrieval problem posed above (Problem~\ref{prob:main}), we give an algorithm (Algorithm~\ref{alg:dcpom}) that separates independent local fitting and calibration from a single round of e-value broadcast followed by a weighted-average fusion at each receiver. The fusion rule combines a per-neighborhood budget $\beta=\alpha/d$ that splits the target among the $d$ agents, a per-agent threshold recalibrated to that budget on a held-out fold, a lifted e-value carrying an observation gate and an uncertainty attenuation, and fixed weights under which fusion is a simple average; the budget lets a single confident agent decide a query in the region it has observed, and we prove (Theorem~\ref{thm:fusion}) that the fused set attains coverage $1-\alpha$ at every agent for any graph and any sensor-noise distribution under local stationarity and temporal mixing rather than i.i.d.\ data, with the attenuation restoring the target level when a local guarantee fails (Proposition~\ref{prop:recovery}). Multi-agent simulations confirm that the bound holds across $\alpha$ and both topologies and that denser graphs leave a smaller unclassified fraction.

The outline of the paper is as follows. The problem formulation is presented in Section \ref{sec:problem formulation} . The Algorithm in presented in Section \ref{sec:algorithm} while its properties are studied in Section \ref{sec:analysis}. The simulation results are presented in Section \ref{sec:simulation} and we conclude in Section  \ref{sec:conclusion}.
\section{Problem Formulation}\label{sec:problem formulation}
In this section we describe the problem setup in subsection \ref{Problem Setup}, the local estimation procedure in \ref{subsec:local_est} and the main problem in subsection \ref{Problem Definition}. All random variables are defined on a common probability space $(\Omega,\mathcal{F},\mathbb{P})$, and the law of a random variable $X$ is denoted $\mathcal{L}(X)$.

\subsection{Problem Setup}\label{Problem Setup}

We consider $N$ robots operating collaboratively. The robots communicate over an undirected graph $\mathcal{G}=(\mathcal{V},\mathcal{E})$, where $\mathcal{V}=\{1,\dots,N\}$ is the set of robots and $(i,j)\in\mathcal{E}$ if robots $i$ and $j$ can exchange messages.  Let $\mathcal{N}(i):=\{j:(i,j)\in\mathcal{E}\}$ denote the neighborhood of robot $i$. Each robot communicates only with its immediate neighbors; there is no central coordinator.The robots operate in a bounded region $E\subset\mathbb{R}^2$ with a deterministic occupancy function $O:E\to\{-1,+1\}$ that is unknown and must be inferred from noisy measurements. Each robot $i$ moves along a trajectory $(e_{0:n}^i,R_{0:n}^i)\in(\mathbb{R}^2\times SO(2))^{n+1}$ and collects 2D LiDAR scans. A scan at time $n$ consists of $L_n^i$ range-bearing returns:
\begin{equation}
Z_n^i = \bigl\{(r_n^{i,\ell},\varphi_n^{i,\ell})\bigr\}_{\ell=1}^{L_n^i},
\end{equation}
where $r_n^{i,\ell}\ge 0$ is the range and $\varphi_n^{i,\ell}$ is the bearing of beam $\ell$ in the robot's local frame.
Both odometry and sensor returns are corrupted by additive Gaussian noise:
\begin{align}
  \tilde{r}_n^{i,\ell} &= r_n^{i,\ell}+\epsilon_{r,n}^{i,\ell}, &
  \tilde{\varphi}_n^{i,\ell} &= \varphi_n^{i,\ell}+\epsilon_{\varphi,n}^{i,\ell},
  \label{eq:noise_lidar}\\
  \tilde{e}_n^i &= e_n^i+\epsilon_{e,n}^{i}, &
  \tilde{\theta}_n^i &= \theta_n^i+\epsilon_{\theta,n}^{i}.
  \label{eq:noise_pose}
\end{align}
where $\epsilon_{r,n}^{i,\ell}\!\sim\!\mathcal{N}(0,\sigma_r^2)$, $\epsilon_{\varphi,n}^{i,\ell}\!\sim\!\mathcal{N}(0,\sigma_\varphi^2)$, $\epsilon_{e,n}^{i}\!\sim\!\mathcal{N}(0,\sigma_e^2 I)$, $\epsilon_{\theta,n}^{i}\!\sim\!\mathcal{N}(0,\sigma_\theta^2)$, and all noise terms are jointly independent. For beam $\ell$, the endpoint is occupied, and points along the ray are free. Using the noisy pose and returns, scan $Z_n^i$ yields labeled samples:
\begin{equation}
\tilde{\mathcal{D}}_n^i = \bigcup_{\ell=1}^{L_n^i}
  \Bigl(\{(\tilde{z}_n^{i,\ell},+1)\}
  \cup\{(\tilde{x}_n^{i,\ell,m},-1)\}_{m=1}^{M_n^{i,\ell}}\Bigr),
\label{eq:scan_to_samples}
\end{equation}
where $\tilde{z}_n^{i,\ell}$ is the noisy beam endpoint and $\tilde{x}_n^{i,\ell,m}$ are free-space samples along the beam. We denote the stochastic process of sampled map locations by $X_n^i\in E$, driven by the robot's noisy pose. For any query $X\in E$, the clean label is $Y:=O(X)\in\{-1,+1\}$, and each sampled location yields an observed noisy label $\tilde{Y}_n^i\in\{-1,+1\}$. The joint distribution of location and noisy label is
\begin{equation}
d\mathbb{P}_{X,\tilde{Y},n}^i(x,\tilde{y})
  = d\mathbb{P}^i\!\bigl(\tilde{Y}_n^i=\tilde{y}\mid X_n^i=x\bigr)\,d\mathbb{P}_{X,n}^i(x),
\label{eq:noisy_label}
\end{equation}
where the conditional measurement distribution is a function of the deterministic $O(x)$. Fix a block half-length $T$.  Partition the time indices into blocks $\mathcal{I}_k := \{(2k-2)T+1,\dots,2kT\}$ for $k=1,2,\dots$, each of length $2T$.  Each block $k$ is split into a training half and a calibration half:
\begin{align*}
\mathcal{I}_k^{\mathrm{train}} &:= \{(2k-2)T+1,\dots,(2k-1)T\},\\
\mathcal{I}_k^{\mathrm{cal}}   &:= \{(2k-1)T+1,\dots,2kT\}.
\end{align*}
Define the training dataset $\mathcal{D}_k^{i,\mathrm{train}}:=\{(X_n^i,\tilde{Y}_n^i)\}_{n\in\mathcal{I}_k^{\mathrm{train}}}$, calibration dataset $\mathcal{D}_k^{i,\mathrm{cal}}:=\{(X_n^i,\tilde{Y}_n^i)\}_{n\in\mathcal{I}_k^{\mathrm{cal}}}$, and training $\sigma$-algebra
$\mathcal{G}_k^i:=\sigma\bigl((X_n^i,\tilde{Y}_n^i):n\in\mathcal{I}_k^{\mathrm{train}}\bigr)$.
Global stationarity is unrealistic in mobile robotics since the robot's spatial sampling distribution evolves as exploration proceeds, and the environment is spatially heterogeneous. The exchangeability assumption in conformal prediction is generally violated; related failures and remedies under distribution shift, covariate shift, and time-series dependence have been studied in \cite{tibshirani2019conformal,gibbs2021aci,zaffran2022adaptive,barber2023conformal}. Following \cite{barber2023conformal,mao2024valid},  we impose local stationarity within each block and temporal mixing across blocks. Distributional closeness is measured in total variation distance $d_{\mathrm{TV}}(P,Q):=\sup_{A}|P(A)-Q(A)|$ \cite{dudley2018real}.

\begin{assumption}[Local Stationarity]\label{ass:stationarity}
For each agent $i$ and block $k$,
\[
\sup_{n,n'\in\mathcal{I}_k^{\mathrm{cal}}}
  d_{\mathrm{TV}}\bigl(\mathcal{L}(X_n^i,\tilde{Y}_n^i),\,
                        \mathcal{L}(X_{n'}^i,\tilde{Y}_{n'}^i)\bigr)\le\delta_T,
\]
where $\delta_T\to 0$ as $T$ grows appropriately.
\end{assumption}

Within each block, the robot explores a locally homogeneous region: occupancy structure, sensor noise, and motion statistics remain approximately constant over the block duration but may change substantially across blocks.

\begin{assumption}[$\beta$-Mixing]\label{ass:mixing}
For each agent $i$, the stochastic process $\{(X_n^i,\tilde{Y}_n^i)\}_{n\ge 1}$ is $\beta^i$-mixing with coefficient $\beta^i(\ell)\to 0$ as $\ell\to\infty$.
\end{assumption}

This captures the decay of temporal dependence due to continuous robot motion. Together, Assumptions~\ref{ass:stationarity} and~\ref{ass:mixing} ensure that the calibration data behaves approximately as an i.i.d.\ sequence within each block.

\subsection{Local Estimation}\label{subsec:local_est}

Given $\mathcal{D}_k^{i,\mathrm{train}}$, agent $i$ fits an ExactGP with a Mat\'{e}rn-$\tfrac{3}{2}$ kernel to obtain a predictive model
\begin{equation}
Y \mid X=x \;\sim\; \mathcal{N}\!\bigl(\hat{\mu}_k^i(x),\,(\hat{\sigma}_k^i)^2(x)\bigr).
\end{equation}

Conditioned on $\mathcal{G}_k^i$, the fitted functions $\hat{\mu}_k^i$ and $\hat{\sigma}_k^i$ are deterministic. For each calibration sample $n\in\mathcal{I}_k^{\mathrm{cal}}$, define the normalized residual nonconformity score
\begin{equation}
A_n^i := \frac{|\tilde{Y}_n^i - \hat{\mu}_k^i(X_n^i)|}{\hat{\sigma}_k^i(X_n^i)}.
\label{eq:nc_score}
\end{equation}
For a query pair $(x^\star,y^\star)$, define $A^i(x^\star,y^\star):=|y^\star-\hat{\mu}_k^i(x^\star)|/\hat{\sigma}_k^i(x^\star)$. The conformal p-value is
\begin{equation}
p_k^i(x^\star,y^\star)
  := \frac{1}{T+1}\!\left(1+\sum_{m\in\mathcal{I}_k^{\mathrm{cal}}}
     \mathbf{1}\{A_m^i\ge A^i(x^\star,y^\star)\}\right).
\label{eq:pval}
\end{equation}
The local prediction set at inner level $a$ is $\Gamma_k^i(x^\star):=\{y:p_k^i(x^\star,y)>a\}$, and $\{p_k^i(X_k^\star,Y_k^\star)\le a\}$ is its miscoverage event; a standard choice is $a=\alpha$, but the fused operator of Section~\ref{subsec:fusion} replaces it with a per-agent recalibrated threshold $a_k^i$. We consider the relevance of agents in two distinct roles. The first is a binary \emph{observation gate} that activates agent $i$ only where it has actually collected data. A natural candidate would be the GP predictive standard deviation $\hat{\sigma}_k^i(x^\star)$, but it is unsuitable for this role as for a stationary Mat\'ern kernel, $\hat{\sigma}_k^i(x)$ saturates at the prior level $\bar{\sigma}^i:=\sqrt{1+(\sigma_n^i)^2}$ (the bare Mat\'ern kernel fixes the signal variance to $1$, and $\sigma_n^i$ is the GP likelihood (observation-noise) standard deviation, fit per agent) as $x$ leaves the explored region, and with a moderate length-scale it can remain well below any fixed threshold over the entire bounded workspace. We instead gate on proximity to the agent's own training inputs,
\begin{equation}
\text{agent $i$ active at $x^\star$}\;\;\Longleftrightarrow\;\;
\operatorname{dist}\!\bigl(x^\star,\{X_n^i\}_{n\in\mathcal{I}_k^{\mathrm{train}}}\bigr)<r, \label{eq:gate}
\end{equation}
where $\operatorname{dist}(x^\star, S)\coloneq\min_{s\in S}\|x^\star-s\|_2$ is the Euclidean distance from $x^\star$ to the nearest point of the set $S$, and $r$ is a radius on the order of the inter-sample spacing. The gate uses the agent's own training inputs $\{X_n^i\}_{n\in\mathcal{I}_k^{\mathrm{train}}}$, so, like $\hat{\mu}_k^i$ and $\hat{\sigma}_k^i$, it is deterministic given $\mathcal{G}_k^i$;  each agent evaluates \eqref{eq:gate} on its own data. The second role is a smooth \emph{attenuation} that down-weights confident-but-uncertain queries; here $\hat{\sigma}_k^i$ is appropriate, and we use the decay \eqref{eq:attenuation}.

\subsection{Problem Definition}\label{Problem Definition}
Let $d_i:=|\mathcal{N}(i)\cup\{i\}|$. We allocate the target miscoverage as a \emph{budget} across the neighborhood,
\begin{equation}
\beta_i := \frac{\alpha}{d_i},
\label{eq:budget}
\end{equation}
and have each contributing agent $j$ recalibrate its detection threshold to that budget. The agents agree on a common set of query locations $\mathcal{Q}_k=\{X^\star_{k,\ell}\}_{\ell=1}^{L_k}$, with $L_k:=|\mathcal{Q}_k|$ (a shared reference grid over the workspace, held fixed across blocks, or a block-specific finite test set), and collaborate by exchanging e-values described as follows. Each element of $\mathcal{Q}_k$ is a realization of the random test location whose block-$k$ law enters the guarantee below; we write $(X_k^\star,Y_k^\star)$ for a generic such location and its true label $Y_k^\star:=O(X_k^\star)$, unobserved at query time. Because $Y_k^\star$ is unknown, each agent evaluates its e-value $e_k^i(X^\star_{k,\ell},y^\star)$ at every query location $X^\star_{k,\ell}\in\mathcal{Q}_k$ and both candidate labels $y^\star\in\{-1,+1\}$. A non-negative scalar $e$ is an e-value if $\mathbb{E}[e]\le 1$; by Markov's inequality, a large e-value is then calibrated evidence against a hypothesis. The use of e-values is motivated by the closure of e-values under deterministic convex averaging: if each local statistic has expectation at most one under the null, then any fixed weighted average remains an e-value \cite{vovk2021evalues,wang2025only}. Recent work on conformal e-prediction shows that conformal prediction sets can be constructed from e-values rather than p-values: for each candidate label \(y\), one computes an e-value ($e(x,y)$) that is valid under the hypothesis \(Y=y\), and then includes \(y\) whenever \(e(x,y)<1/\alpha\),  \cite{gauthier2025values,vovk2025conformal}. 

At the end of block $k$, for every $X^\star_{k,\ell}\in\mathcal{Q}_k$ and each label $y^\star\in\{-1,+1\}$, agent $i$ computes an e-value $e_k^i(X^\star_{k,\ell},y^\star)$ from its local model to measure the evidence against $y^\star$ as follows. Writing $\epsilon_k^i(a):=\mathbb{P}\bigl(p_k^i(X_k^\star,Y_k^\star)\le a\bigr)$ for the local miscoverage of the inner set $\{y:p_k^i(x,y)>a\}$, agent $i$ sets
\begin{equation}
a_k^i := \sup\{a\in(0,1):\epsilon_k^i(a)\le\beta_i\},
\; \epsilon_k^i := \epsilon_k^i(a_k^i)\le\beta_i ,
\label{eq:recalib}
\end{equation}
estimated on a held-out fold (Remark~\ref{rem:epsilon}). For a query $(X^\star_{k,\ell},y^\star)$, agent $i$ defines the lifted violation indicator:
\begin{align}
e_k^i(X^\star_{k,\ell},y^\star)
&:= g_k^i(X^\star_{k,\ell})\,\frac{\mathbf{1}\{p_k^i(X^\star_{k,\ell},y^\star)\le a_k^i\}}{\epsilon_k^i} \nonumber\\
&\quad \cdot \mathbf{1}\{\operatorname{dist}(X^\star_{k,\ell},\{X_n^i\}_{n\in\mathcal{I}_k^{\mathrm{train}}})<r\},
\label{eq:evalue_def}
\end{align}
where the gate $\mathbf{1}\{\operatorname{dist}(X^\star_{k,\ell},\{X_n^i\}_{n\in\mathcal{I}_k^{\mathrm{train}}})<r\}$ of \eqref{eq:gate} activates agent $i$ only where it has observed data: this confines each agent's contribution to the region where its local-coverage bound applies, so it abstains ($e_k^i=0$) rather than casting an uncalibrated vote where it has no data. The factor $g_k^i(X^\star_{k,\ell})\in(0,1]$ is the attenuation
\begin{equation}
g_k^i(X^\star_{k,\ell}) := \exp\!\bigl(-\lambda\, \hat{\sigma}_k^i(X^\star_{k,\ell})\bigr)\in(0,1],\qquad \lambda\ge 0,
\label{eq:attenuation}
\end{equation}
which further down-weights queries where the agent is uncertain, supplying the margin that recovers coverage when the local bound is loose (Proposition~\ref{prop:recovery}); $\lambda=0$ recovers the unattenuated indicator. The attenuation is placed inside the e-value rather than in the mixing weights: since $g_k^i(X^\star_{k,\ell})\le 1$ pointwise, it can only shrink each agent's contribution, so the e-value property $\mathbb{E}[e_k^i]\le 1$ is preserved (see the proof of Theorem~\ref{thm:fusion}). Thus,  the e-value is constructed locally using the conformal p-value \eqref{eq:pval}, an observation gate, an attenuation factor, and a recalibrated threshold. Agent $i$ then broadcasts only these scalars to its neighbors, ensuring no raw data is ever shared.

\begin{problem}\label{prob:main}
Fix the communication graph $\mathcal{G}=(\mathcal{V},\mathcal{E})$ of Section~\ref{Problem Setup}. At the end of block $k$, every agent $j\in\mathcal{V}$ computes its local e-values $e_k^j(X^\star_{k,\ell},y^\star)$, $\ell=1,\dots,L_k$, of \eqref{eq:evalue_def} over the shared query set $\mathcal{Q}_k$ and transmits them to its neighbors $\mathcal{N}(j)$, so that each agent $i$ receives $\{e_k^j\}_{j\in\mathcal{N}(i)\cup\{i\}}$ in one communication round. Let $(X_k^\star,Y_k^\star)$ be a test point whose law is close, in total variation, to the block-$k$ calibration law of every active agent $j\in\mathcal{N}(i)\cup\{i\}$:
\[
\sup_{n\in\mathcal{I}_k^{\mathrm{cal}}} d_{\mathrm{TV}}\bigl(\mathcal{L}(X_n^j,\tilde{Y}_n^j),\,\mathcal{L}(X_k^\star,\tilde{Y}_k^\star)\bigr)\le\delta_T .
\]
The objective is to design, at each agent $i$, a fusion rule that combines the received e-values $\{e_k^j\}_{j\in\mathcal{N}(i)\cup\{i\}}$ into a fused prediction set $\Gamma_k^{\mathrm{fused},i}(X^\star_{k,\ell})$ at each query location, satisfying
\[
\mathbb{P}\!\bigl(Y_k^\star\notin\Gamma_k^{\mathrm{fused},i}(X_k^\star)\bigr)\le\alpha.
\]
\end{problem}

The key difficulty is that temporal correlation in the robot's data violates the i.i.d.\ exchangeability assumption on which split conformal prediction relies, so a single agent can no longer certify that its local miscoverage remains at $\alpha$. Compounding this, any one agent observes only part of the workspace. Fusion over the neighborhood, together with the per-neighborhood budget introduced in Section~\ref{subsec:fusion}, recovers the target level $\alpha$ while letting a lone confident agent still decide a query in its own region, without requiring i.i.d.\ data or any additional structural assumptions.
\section{Algorithm}\label{sec:algorithm}
In this section, we present the fusion operator which solves Problem \ref{prob:main} and then subsequently  an algorithm utilizing the same for occupancy map estimation. 
\subsection{Fusion Operator}\label{subsec:fusion}
Define uniform weights over the neighborhood:
\begin{equation}
w_k^j := \frac{1}{|\mathcal{N}(i)\cup\{i\}|}, \quad j\in\mathcal{N}(i)\cup\{i\},
\label{eq:weights}
\end{equation}
so that $\sum_{j\in\mathcal{N}(i)\cup\{i\}}w_k^j = 1$. Keeping the mixing weights fixed and data-independent is essential for validity: query-dependent convex weights that renormalize over the active set would inflate the contributions of surviving agents and require coverage conditional on the query location, which split conformal prediction does not provide. For each $X^\star_{k,\ell}\in\mathcal{Q}_k$, the aggregated e-value, fused p-value, and fused prediction set are:
\begin{align}
&e_k^{\mathrm{fused},i}(X^\star_{k,\ell},y^\star)
  := \sum_{j\in\mathcal{N}(i)\cup\{i\}} w_k^j\,e_k^j(X^\star_{k,\ell},y^\star),
\label{eq:efused}\\
&p_k^{\mathrm{fused},i}(X^\star_{k,\ell},y^\star)
  := \min\!\left(1,\,\frac{1}{e_k^{\mathrm{fused},i}(X^\star_{k,\ell},y^\star)}\right),
\label{eq:pfused}\\
&\Gamma_k^{\mathrm{fused},i}(X^\star_{k,\ell})
  := \bigl\{y\in\{-1,+1\}:p_k^{\mathrm{fused},i}(X^\star_{k,\ell},y)>\alpha\bigr\}.
\label{eq:gamma_fused}
\end{align}
The maps $e_k^{\mathrm{fused},i}$, $p_k^{\mathrm{fused},i}$, and $\Gamma_k^{\mathrm{fused},i}$ are deterministic given the fitted models $\{\hat{\mu}_k^j,\hat{\sigma}_k^j\}_{j}$; the randomness in the coverage statements of Section~\ref{sec:analysis} arises solely from evaluating them at the random test point $(X_k^\star,Y_k^\star)$.

\begin{remark}
When no agent in $\mathcal{N}(i)\cup\{i\}$ is active at $X^\star_{k,\ell}$ (all $\operatorname{dist}(X^\star_{k,\ell},\{X_n^j\}_{n\in\mathcal{I}_k^{\mathrm{train}}})\ge r$),
every $e_k^j=0$, so $e_k^{\mathrm{fused},i}=0$ and $p_k^{\mathrm{fused},i}=1$. Thus $\Gamma_k^{\mathrm{fused},i}(X^\star_{k,\ell})=\{-1,+1\}$; the prediction set is trivially valid but uninformative. This is the correct conservative behavior when no agent has observed the queried region.
\end{remark}

\subsection{Coverage Recovery via Fusion}

Algorithm~\ref{alg:dcpom} summarizes the full procedure. Phases~1--2 run independently at each agent with no communication beyond agreeing on the query set $\mathcal{Q}_k$. In Phase~2, each agent broadcasts only its scalar e-values over $\mathcal{Q}_k$; no raw data, GP model, or calibration scores leave the agent. Phase~3 fuses the received e-values by a weighted average. The payload is $O(|\mathcal{Q}_k|)$ per agent and sparse, since $e_k^i$ vanishes wherever agent $i$ is inactive or does not flag the candidate label.

The budget \eqref{eq:budget} controls how decisive the fusion is. Because each e-value is normalized by its own miscoverage $\epsilon_k^j\approx\beta_i=\alpha/d_i$, a single active agent that flags the label contributes $e_k^j\approx d_i/\alpha$; after the $1/d_i$ averaging in the fused e-value \eqref{eq:efused}, this reaches the fused decision threshold $1/\alpha$. Without the budget, each agent would instead normalize by $\alpha$, so one flagging agent would contribute only $1/(\alpha d_i)$ after averaging, and a label could be excluded only if several agents flagged it together. The budget removes this factor $d_i$, so an agent can decide queries in the region it has observed on its own.

\begin{algorithm}[t]
\caption{Conformal Coverage Recovery via Fusion}
\label{alg:dcpom}
\begin{algorithmic}[1]
\Require Agents $i=1,\dots,N$; block data $\{(X_n^i,\tilde{Y}_n^i)\}_{n\in\mathcal{I}_k}$;
         $d$-regular graph $\mathcal{G}$ ($d=|\mathcal{N}(i)\cup\{i\}|$); level $\alpha$; gate radius $r$; decay rate $\lambda$; query set $\mathcal{Q}_k$
\Ensure Fused prediction set $\Gamma_k^{\mathrm{fused},i}(X^\star_{k,\ell})$, $X^\star_{k,\ell}\in\mathcal{Q}_k$, at each agent $i$

\Statex \textbf{Phase 1: Local Calibration and Recalibration} (each agent $i$)
\For{each agent $i=1,\dots,N$}
  \State $(\hat{\mu}_k^i,\hat{\sigma}_k^i)\leftarrow\mathrm{TrainGP}(\mathcal{D}_k^{i,\mathrm{train}})$
  \State $A_n^i\leftarrow|\tilde{Y}_n^i-\hat{\mu}_k^i(X_n^i)|/\hat{\sigma}_k^i(X_n^i)$,\;\;$n\in\mathcal{I}_k^{\mathrm{cal}}$
    \Comment{\eqref{eq:nc_score}}
  \State $\beta\leftarrow\alpha/d$;\;\; recalibrate $a_k^i$ and $\epsilon_k^i$ on a held-out fold
    \Comment{\eqref{eq:budget},\eqref{eq:recalib}}
\EndFor

\Statex \textbf{Phase 2: Local E-values and Broadcast} (each agent $i$)
\For{each $X^\star_{k,\ell}\in\mathcal{Q}_k$ and $y^\star\in\{-1,+1\}$}
  \State $A^i\leftarrow|y^\star-\hat{\mu}_k^i(X^\star_{k,\ell})|/\hat{\sigma}_k^i(X^\star_{k,\ell})$;\;\;
         $p_k^i\leftarrow\frac{1}{T+1}\bigl(1+\#\{m:A_m^i\ge A^i\}\bigr)$
    \Comment{\eqref{eq:pval}}
  \State $e_k^i(X^\star_{k,\ell},y^\star)\leftarrow
    \begin{aligned}[t]
    & e^{-\lambda\hat{\sigma}_k^i(X^\star_{k,\ell})}
      \frac{\mathbf{1}\{p_k^i\le a_k^i\}}{\epsilon_k^i} \\
    & \cdot \mathbf{1}\{\operatorname{dist}(X^\star_{k,\ell},\{X_n^i\}_{n\in\mathcal{I}_k^{\mathrm{train}}})<r\}
    \end{aligned}$
    \Comment{\eqref{eq:evalue_def}}
\EndFor
\State Broadcast $\{e_k^i(X^\star_{k,\ell},\cdot)\}_{X^\star_{k,\ell}\in\mathcal{Q}_k}$ to all $j\in\mathcal{N}(i)$

\Statex \textbf{Phase 3: Fusion} (each agent $i$, on received e-values)
\For{each $X^\star_{k,\ell}\in\mathcal{Q}_k$ and $y^\star\in\{-1,+1\}$}
  \State $e_k^{\mathrm{fused},i}\leftarrow
    \tfrac{1}{d}\sum_{j\in\mathcal{N}(i)\cup\{i\}}e_k^j(X^\star_{k,\ell},y^\star)$;\;\;
    $p_k^{\mathrm{fused},i}\leftarrow\min(1,\,1/e_k^{\mathrm{fused},i})$
    \Comment{\eqref{eq:efused},\eqref{eq:pfused}}
\EndFor
\State \Return $\Gamma_k^{\mathrm{fused},i}(X^\star_{k,\ell})\leftarrow
  \{y^\star:p_k^{\mathrm{fused},i}(X^\star_{k,\ell},y^\star)>\alpha\}$
  \Comment{\eqref{eq:gamma_fused}}
\end{algorithmic}
\end{algorithm}

\begin{remark}\label{rem:epsilon}

The threshold $a_k^i$ and its normalizer $\epsilon_k^i$ in \eqref{eq:recalib} are computed from a held-out fold: a labeled set the agent holds that is disjoint from the data used to fit its GP and to compute its calibration scores, and that follows the same local regime as the test point. On this fold, $a_k^i$ is the largest inner level at which the agent's measured miscoverage stays within the budget $\beta_i=\alpha/d_i$, and $\epsilon_k^i$ is the miscoverage measured at $a_k^i$. Keeping this fold disjoint from the calibration scores makes $a_k^i$ and $\epsilon_k^i$ independent of the scores to which they are applied. The choice of this fold in our experiments is the calibration half of the following block, as discussed in Section~\ref{subsec:sim_setup}.
The validity proof (Theorem~\ref{thm:fusion}) needs only one property of $\epsilon_k^i$: that it upper-bounds the true miscoverage $\mathbb{P}(p_k^i(X_k^\star,Y_k^\star)\le a_k^i)$. The held-out estimate provides this under the temporal protocol of Section~\ref{subsec:local_est}. We require $\epsilon_k^i>0$, which holds for finite $T$ because a p-value cannot fall below $1/(T{+}1)$. When the attenuation is tuned as in Proposition~\ref{prop:recovery}, $\epsilon_k^i$ may instead be set to the smaller confident-interior miscoverage $\epsilon_0^i$.

\end{remark}

\section{Analysis}\label{sec:analysis}

We establish that Algorithm~\ref{alg:dcpom} solves Problem~\ref{prob:main}.

\begin{assumption}[Local Coverage]\label{ass:local_coverage}
For each agent $i$, block $k$, and the recalibrated inner threshold $a_k^i$ of \eqref{eq:recalib}, the denominator $\epsilon_k^i\in(0,1)$ upper-bounds the local miscoverage of the inner set $\{y:p_k^i(x,y)>a_k^i\}$: for every test point $(X_k^\star,Y_k^\star)$ whose law satisfies
\begin{align}\label{eq:teststar}
\sup_{n\in\mathcal{I}_k^{\mathrm{cal}}}d_{\mathrm{TV}}\bigl(\mathcal{L}(X_n^i,\tilde{Y}_n^i),\,\mathcal{L}(X_k^\star,\tilde{Y}_k^\star)\bigr)\le\delta_T, \nonumber \\
\implies  \mathbb{P}\!\bigl(p_k^i(X_k^\star,Y_k^\star)\le a_k^i\bigr)\le\epsilon_k^i .
\end{align}
\end{assumption}
The above assumption is in principle the same as the results in \cite{barber2023conformal,mao2024valid}. 
The joint condition \eqref{eq:teststar} refers to the law of $(X_k^\star,\tilde Y_k^\star)$, yet the label is unavailable at query time. The next result shows that, under the shared measurement kernel of our sensing model, it can nonetheless be certified from the spatial marginal $\mathcal{L}(X)$ alone.

\begin{proposition}[Testing the marginal suffices]\label{prop:marginal}
Fix agent $i$ and block $k$, and suppose the noisy label is generated by a common measurement kernel: there is a Markov kernel $K(\cdot\mid x)$ on $\{-1,+1\}$ such that
\begin{align*}
\mathcal{L}(\tilde Y_n^i\mid X_n^i=x)&=K(\cdot\mid x)\;\;\text{for all }n\in\mathcal{I}_k^{\mathrm{cal}}, \\
\mathcal{L}(\tilde Y_k^\star\mid X_k^\star=x)&=K(\cdot\mid x).
\end{align*}

Then for every $n\in\mathcal{I}_k^{\mathrm{cal}}$,
\[
d_{\mathrm{TV}}\bigl(\mathcal{L}(X_n^i,\tilde Y_n^i),\,\mathcal{L}(X_k^\star,\tilde Y_k^\star)\bigr)
= d_{\mathrm{TV}}\bigl(\mathcal{L}(X_n^i),\,\mathcal{L}(X_k^\star)\bigr).
\]
Consequently, the joint closeness \eqref{eq:teststar} holds if and only if the spatial marginals are within $\delta_T$, so any relevance signal that is a function of $x$ alone, such as the observation gate \eqref{eq:gate} or the predictive uncertainty $\hat{\sigma}_k^i(x^\star)$, controls the quantity Assumption~\ref{ass:local_coverage} constrains.
\end{proposition}
\textit{Proof.}
Write $\mu=\mathcal{L}(X_n^i)$ and $\nu=\mathcal{L}(X_k^\star)$, and let $P,Q$ be the joint laws, which by the shared-kernel hypothesis disintegrate through the same $K$:
$P(dx,dy)=K(dy\mid x)\,\mu(dx)$ and $Q(dx,dy)=K(dy\mid x)\,\nu(dx)$.

\emph{Lower bound.} The coordinate projection $\pi(x,y)=x$ is measurable and pushes $P,Q$ forward to $\mu,\nu$. Since a measurable map cannot increase total variation \cite{dudley2018real},
$d_{\mathrm{TV}}(P,Q)\ge d_{\mathrm{TV}}(\pi_\#P,\pi_\#Q)=d_{\mathrm{TV}}(\mu,\nu)$.

\emph{Upper bound.} Let $\lambda=\mu-\nu$, a finite signed measure with $\lambda(\mathcal{X})=0$. For measurable $A\subseteq\mathcal{X}\times\{-1,+1\}$ with sections $A_x=\{y:(x,y)\in A\}$,
\[
P(A)-Q(A)=\int_{\mathcal{X}} K(A_x\mid x)\,\lambda(dx),
\]
and $x\mapsto K(A_x\mid x)$ is measurable with values in $[0,1]$. Hence
\[
|P(A)-Q(A)|\le\sup_{0\le h\le 1}\int h\,d\lambda=\lambda_+(\mathcal{X})=d_{\mathrm{TV}}(\mu,\nu),
\]
where $\lambda_+(\mathcal{X})=d_{\mathrm{TV}}(\mu,\nu)$ because $\lambda(\mathcal{X})=0$. Taking the supremum over $A$ gives $d_{\mathrm{TV}}(P,Q)\le d_{\mathrm{TV}}(\mu,\nu)$, and the two bounds yield equality.\QEDn\\
The shared-kernel hypothesis holds in our sensing model because the occupancy $O(x)$ is deterministic and the noise parameters in \eqref{eq:noise_lidar}--\eqref{eq:noise_pose} are location- and index-invariant, so $\mathcal{L}(\tilde Y\mid X=x)$ depends on $x$ alone.

\subsection{E-values and the Markov Inequality}

The fused coverage guarantee rests on the theory of e-values \cite{vovk2021evalues}. Recall from Section~\ref{sec:problem formulation} that a non-negative random variable $e$ is an e-value if $\mathbb{E}[e]\le 1$. The following standard property is used by the analysis.

\begin{lemma}[Markov Inequality for E-values]
\label{lem:markov}
If $e$ is an e-value, then for any $\alpha\in(0,1)$,$\mathbb{P}(e\ge 1/\alpha)\le\alpha.$Equivalently, for $p_e:=\min(1,1/e)$, we have $\mathbb{P}(p_e\le\alpha)\le\alpha$.
\end{lemma}
\textit{Proof.}
Direct application of Markov's inequality:$\mathbb{P}(e\ge 1/\alpha)\le\alpha\,\mathbb{E}[e]\le\alpha.$
The equivalence follows from $\{e\ge 1/\alpha\}=\{1/e\le\alpha\}=\{p_e\le\alpha\}$.\QEDn
\begin{theorem}[Fused Coverage]\label{thm:fusion}
Under Assumption~\ref{ass:local_coverage}, for any agent $i$, block $k$, and test point $(X_k^\star,Y_k^\star)$:
\[
\mathbb{P}\!\bigl(Y_k^\star\notin\Gamma_k^{\mathrm{fused},i}(X_k^\star)\bigr)\le\alpha.
\]
\end{theorem}
\textit{Proof.}
Fix agent $i$ and block $k$. We show $e_k^{\mathrm{fused},i}$ is an e-value and apply Lemma~\ref{lem:markov}. Note that each $e_k^j$ is an e-value and that is proven as follows. 
For any $j\in\mathcal{N}(i)\cup\{i\}$,
\begin{align*}
\mathbb{E}[e_k^j] \hspace{-3pt}
  &=  \hspace{-2pt}\frac{1}{\epsilon_k^j}\,
     \mathbb{E}\!\bigl[g_k^j(X_k^\star)\,\mathbf{1}\{p_k^j\le a_k^j\}\cdot\mathbf{1}\{\operatorname{dist}(X_k^\star,\{X_n^j\})<r\}\bigr]\\
  &\le \frac{1}{\epsilon_k^j}\,
     \mathbb{E}\!\bigl[\mathbf{1}\{p_k^j\le a_k^j\}\bigr]
   = \frac{1}{\epsilon_k^j}\,
     \mathbb{P}\!\bigl(p_k^j(X_k^\star,Y_k^\star)\le a_k^j\bigr)
   \le 1,
\end{align*}
where the first inequality uses $0\le g_k^j(X_k^\star)\le 1$ and the gate indicator $\le 1$, and the last uses Assumption~\ref{ass:local_coverage}. Hence $\mathbb{E}[e_k^j]\le 1$ for all $j\in\mathcal{N}(i)\cup\{i\}$. By linearity of expectation and $\sum_j w_k^j = 1$:
\[
\mathbb{E}\bigl[e_k^{\mathrm{fused},i}\bigr]
  = \sum_{j\in\mathcal{N}(i)\cup\{i\}} w_k^j\,\mathbb{E}[e_k^j]
  \le \sum_{j\in\mathcal{N}(i)\cup\{i\}} w_k^j = 1,
\]
thus $e_k^{\mathrm{fused},i}$ is an e-value. Invoking Lemma~\ref{lem:markov},  it follows that
\begin{align*}
\mathbb{P}\!\bigl(Y_k^\star\notin\Gamma_k^{\mathrm{fused},i}(X_k^\star)\bigr)
  &= \mathbb{P}\!\bigl(p_k^{\mathrm{fused},i}(X_k^\star,Y_k^\star)\le\alpha\bigr)\\
  &= \mathbb{P}\!\bigl(e_k^{\mathrm{fused},i}(X_k^\star,Y_k^\star)\ge 1/\alpha\bigr)
   \le \alpha.
\end{align*}\QEDn

\subsection{Coverage Recovery via Uncertainty Attenuation}

Theorem~\ref{thm:fusion} assumes the local-coverage bound holds for every active agent. In practice, an agent may be active at $x^\star$ (it has observed data nearby) yet not strictly compliant, so its conditional local miscoverage at $x^\star$ exceeds $\epsilon_k^j$. We now show that the attenuation $g_k^j$ recovers the target level under a graded control on how miscoverage degrades with predictive uncertainty.

\begin{assumption}[Uncertainty-graded miscoverage]\label{ass:profile}
For each agent $i$ and block $k$ there exist $\epsilon_0^i\in(0,1)$ and a growth rate $\mu^i\ge 0$ such that the conditional local miscoverage
\[
c_k^i(x):=\mathbb{P}\!\bigl(p_k^i(X_k^\star,Y_k^\star)\le a_k^i\,\big|\,X_k^\star=x\bigr)
\]
satisfies $c_k^i(x)\le\epsilon_0^i\,\exp\!\bigl(\mu^i\,\hat{\sigma}_k^i(x)\bigr)$ for every active $x$ (i.e.\ $\operatorname{dist}(x,\{X_n^i\}_{n\in\mathcal{I}_k^{\mathrm{train}}})<r$).
\end{assumption}

This posits that conditional miscoverage is at most $\epsilon_0^i$ in the confident interior and degrades at most exponentially as predictive uncertainty grows. It is a \emph{modeling assumption}, not a derived fact: it holds when the GP variance is a faithful proxy for predictive error, and $\mu^i$ quantifies the residual variance miscalibration in the extrapolation regime.

\begin{proposition}[Recovery]\label{prop:recovery}
Set the e-value denominators to $\epsilon_k^j=\epsilon_0^j$ and let $g_k^j(x^\star)=\exp(-\lambda\hat{\sigma}_k^j(x^\star))$ with $\lambda\ge\max_j\mu^j$. Then under Assumption~\ref{ass:profile}, each $e_k^j$ is an e-value and
\[
\mathbb{P}\!\bigl(Y_k^\star\notin\Gamma_k^{\mathrm{fused},i}(X_k^\star)\bigr)\le\alpha,
\]
without requiring the local-coverage bound (Assumption~\ref{ass:local_coverage}) to hold uniformly over active agents.
\end{proposition}
\textit{Proof.} Conditioning on $X_k^\star=x$ and the calibration data, the gate and attenuation are deterministic, so
\begin{align*}
\mathbb{E}\!\bigl[e_k^j\mid X_k^\star=x ,& \mathcal{G}^i_k\bigr]
  = \\
  &\frac{g_k^j(x)\,\mathbf{1}\{\operatorname{dist}(x,\{X_n^j\}_{n\in\mathcal{I}_k^{\mathrm{train}}})<r\}}{\epsilon_0^j}\,c_k^j(x).    
\end{align*}
On the active region $\{\operatorname{dist}(x,\{X_n^j\}_{n\in\mathcal{I}_k^{\mathrm{train}}})<r\}$, Assumption~\ref{ass:profile} and $\lambda\ge\mu^j$ give
\[
g_k^j(x)\,c_k^j(x)\le\epsilon_0^j\,\exp\!\bigl((\mu^j-\lambda)\hat{\sigma}_k^j(x)\bigr)\le\epsilon_0^j,
\]
so $\mathbb{E}[e_k^j\mid X_k^\star=x]\le 1$ pointwise, hence $\mathbb{E}[e_k^j]\le 1$. Steps~2--3 of the proof of Theorem~\ref{thm:fusion} then apply verbatim. \QEDn

\begin{remark}[Validity--efficiency dial]\label{rem:dial}
Because $\hat{\sigma}_k^j$ is bounded by the prior level $\bar{\sigma}^j=\sqrt{1+(\sigma_n^j)^2}$, the condition $\lambda\ge\mu^j$ is always met with a finite decay rate; taking $\lambda$ large enough that the floor $\exp(-\lambda\bar{\sigma}^j)\le\epsilon_0^j$ controls even the worst case $c_k^j\equiv 1$. Larger $\lambda$ makes recovery easier but shrinks every $g_k^j$, yielding more conservative, larger prediction sets, so $\lambda$ is the validity--efficiency dial. The guarantee remains marginal: Assumption~\ref{ass:profile} is used only to bound the marginal expectation $\mathbb{E}[e_k^j]$, not to assert conditional coverage of the output set.
\end{remark}

\begin{remark}[Topology and the Budget]\label{rem:topology}
Theorem~\ref{thm:fusion} holds for \emph{any} graph $\mathcal{G}$: the factor $1/d_i$ in $e_k^{\mathrm{fused},i}$ and the $d_i$-fold sum over the neighborhood cancel, so $\mathbb{E}[e_k^{\mathrm{fused},i}]\le 1$ regardless of topology and the coverage guarantee is graph-independent. Topology instead shapes efficiency as a denser graph lets each agent reach the data of more peers, so more of the workspace is covered by at least one active neighbor and the no-data fraction shrinks; since the per-neighborhood budget $\beta_i=\alpha/d_i$ keeps each active agent decisive regardless of $d_i$, a denser graph decides more queries. Section~\ref{sec:simulation} confirms that the mesh attains a far smaller no-data fraction and a higher classified fraction than the ring at equal coverage. The budget $\beta_i=\alpha/d_i$ requires each agent to know the (common) neighborhood size, so the clean broadcast form assumes a regular graph; for an irregular graph, one may use the global bound $\beta=\alpha/\max_i d_i$, preserving validity at a mild efficiency cost.
\end{remark}

\section{Simulation}\label{sec:simulation}

We instantiate Algorithm~\ref{alg:dcpom} on a 2D multi-robot occupancy mapping benchmark to validate Theorem~\ref{thm:fusion} and to quantify the effect of communication topology anticipated in Remark~\ref{rem:topology}. Per that remark, validity holds for any graph $\mathcal{G}$; the simulation isolates the two finite-sample consequences of topology that the bound does not capture: the \emph{efficiency} of the fused prediction set (no-data fraction and mean set size) and the \emph{realized} empirical miscoverage relative to the nominal level~$\alpha$.

\subsection{Setup}\label{subsec:sim_setup}

The workspace is a rectangular region $E=[0,30]\times[0,20]\,$m with interior walls forming corridors and obstacles. $N=5$ ground robots, each equipped with a $180^\circ$ planar LiDAR with range $5$\,m, traverse partially overlapping sub-regions of $E$. Sensor and pose noise follow \eqref{eq:noise_lidar}--\eqref{eq:noise_pose} with $\sigma_r=2\!\times\!10^{-2}$\,m, $\sigma_\varphi=3\!\times\!10^{-3}$\,rad, $\sigma_e=4\!\times\!10^{-2}$\,m, $\sigma_\theta=3\!\times\!10^{-3}$\,rad. Each LiDAR scan is converted to a labeled point cloud via \eqref{eq:scan_to_samples} using $5$ free-space samples per beam. Fig.~\ref{fig:data} shows the five trajectories; each agent covers $52$--$62\%$ of the workspace within LiDAR range, in partially overlapping sub-regions.
\begin{figure}[t]
  \centering
  \includegraphics[width=0.9\columnwidth]{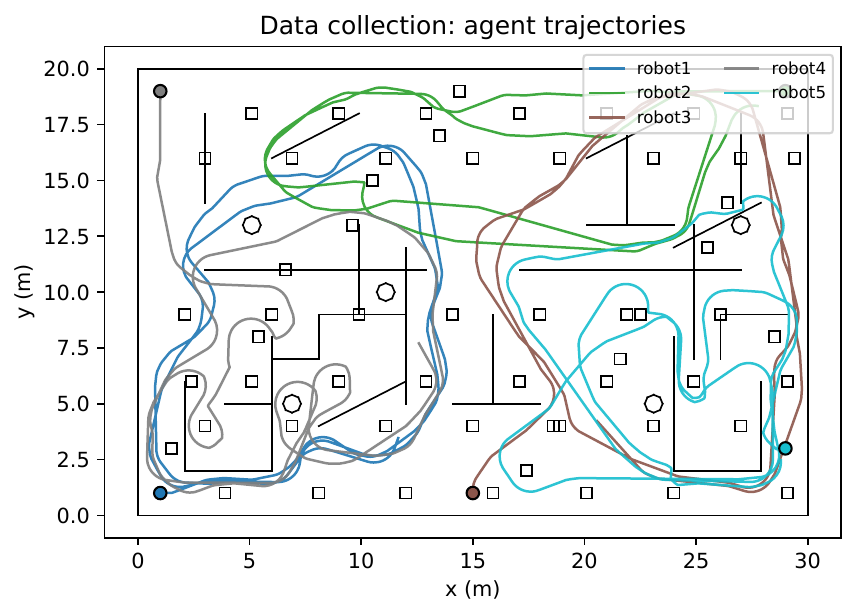}
  \caption{The five agents' trajectories over the workspace $E$ (start markers circled). The agents traverse partially overlapping sub-regions}
  \label{fig:data}
    \vspace{-0.7cm}
\end{figure}
Each agent's trajectory is divided into $K=2$ temporal blocks of equal length, and each block is split into a training half and a calibration half as in Section~\ref{subsec:local_est}. Within each block, agent~$i$ fits an ExactGP with Mat\'ern-$\tfrac{3}{2}$ kernel on up to $1.2\!\times\!10^4$ subsampled labeled points from $\mathcal{D}_k^{i,\mathrm{train}}$ and computes the calibration scores $\{A_n^i\}$ on $\mathcal{D}_k^{i,\mathrm{cal}}$ (capped at $2{,}000$ points). The recalibrated inner threshold $a_k^i$ and denominator $\epsilon_k^i$ of \eqref{eq:recalib} are estimated on a separate held-out fold: the calibration half of the next block, which lies in the same shifted regime as the test set yet is disjoint from it, so the reported coverage is not circular (Remark~\ref{rem:epsilon}). Following the temporal-mixing rationale of Section~\ref{subsec:local_est}, the test set for block $k$ is the training fold of block $k{+}1$, capped at $2{,}000$ points per agent. The union of these per-agent test sets across all $N$ agents forms the shared evaluation set used to compute fused coverage; per-agent coverage is evaluated on each agent's own test fold. The nominal miscoverage level is $\alpha=0.2$. The observation gate \eqref{eq:gate} uses radius $r=1\,$m, on the order of the radial spacing of the LiDAR free-space samples. The per-neighborhood budget is $\beta_i=s\,\alpha/d_i$ with safety factor $s=0.7$, so a lone active agent clears the fused threshold with margin. The attenuation decay is $\lambda=0$ at the operating point and is swept separately to trace the validity--efficiency dial. For the validity study we sweep $\alpha\in\{0.05,0.10,\dots,0.40\}$.

\subsection{Communication Topologies}\label{subsec:topologies}

Ring and mesh topologies have been used in multi-robot occupancy map estimation / SLAM for e.g. \cite{zhang2025flykites,martins2020mrgs}.
\begin{itemize}
\item \textbf{Ring} $\mathcal{G}_{\mathrm{ring}}$: the agents are arranged in a single cycle ordered by spatial adjacency, $1\text{-}2\text{-}3\text{-}5\text{-}4\text{-}1$. Each agent is adjacent to its two cyclic neighbors, so $|\mathcal{N}(i)\cup\{i\}|=3$ for every $i$;
\item \textbf{Mesh} $\mathcal{G}_{\mathrm{mesh}}$: the complete graph $K_N$, so $|\mathcal{N}(i)\cup\{i\}|=N=5$ for every $i$.
\end{itemize}
For each topology and each agent $i$, the fusion of Section~\ref{subsec:fusion} is applied over $\mathcal{N}(i)\cup\{i\}$ on the broadcast e-values, with the budget $\beta_i=s\,\alpha/d_i$ set by the common degree $d_i\in\{3,5\}$. The fused prediction set $\Gamma_k^{\mathrm{fused},i}(x^\star)$ is evaluated at every test point and on a dense grid over $E$ for the map figures.

\subsection{Metrics}\label{subsec:metrics}

We report, per agent, and then averaged across agents:
\begin{itemize}
\item \emph{Coverage} $1-\widehat{m}_i$, the fraction of test points whose true label lies in $\Gamma_k^{\mathrm{fused},i}$;
\item \emph{Classified fraction}, the fraction of singleton sets $|\Gamma|=1$ (a decided free/occupied cell), and \emph{accuracy on classified}, the fraction of those singletons that are correct;
\item \emph{Uncertain fraction}, the fraction of $\Gamma=\{-1,+1\}$ returned \emph{with} data (a genuine conformal abstention), kept separate from the \emph{no-data fraction} $\nu_i$, where no neighbor of $i$ is active and $\Gamma=\{-1,+1\}$ is returned trivially;
\item \emph{Empty fraction}, the fraction of $\Gamma=\varnothing$ (both labels excluded --- the only genuine coverage failure), and the \emph{mean set size} $\bar{s}_i\in[0,2]$.
\end{itemize}
Category fractions are reported both on the whole grid and on the explored region ($\ge 1$ active agent), so trivial no-data sets are never conflated with informative decisions. For the validity study, we plot $1-\widehat{m}_i$ against the nominal coverage $1-\alpha$ for each topology, alongside the ideal diagonal.

\subsection{Implementation Details}\label{subsec:impl}

The pipeline is implemented in Python using GPyTorch for the GPs and a custom NumPy implementation of the conformal, gating, and e-value primitives. Each agent's GP is trained once per block; the distance gate, conformal p-values, recalibration, and e-value fusion are vectorized, so evaluating a second topology on the same broadcast e-values is a single $\mathcal{O}(N\,|\mathcal{Q}_k|)$ pass.

\subsection{Results}\label{subsec:results}

Table~\ref{tab:topology_op} reports the per-agent statistics at $\alpha=0.2$. Both topologies satisfy the fused bound with a large margin: every agent's coverage $1-\widehat{m}_i$ exceeds the target $1-\alpha=0.8$, with agent-averaged values $0.990$ for the ring and $0.976$ for the mesh. The method commits to a singleton only when the e-value evidence is strong, achieving $97$--$99\%$ accuracy on the cells it decides, so realized coverage sits well above the nominal guarantee, and $\alpha$ trades the decided fraction rather than the coverage (Fig.~\ref{fig:coverage_alpha}).As expected, denser communication is better. The mesh decides $84.8\%$ of test queries with only a $1.7\%$ no-data fraction, versus $43.6\%$ decided and $17.4\%$ no-data for the ring, at comparable coverage and accuracy ($0.982$ vs.\ $0.977$). This is due to the per-neighborhood budget $\beta_i=s\alpha/d_i$, as it keeps a lone confident agent decisive even in the mesh's five-agent neighborhood, so the mesh's broader reach reduces no-data without diluting decisions.
\begin{table}[ht]
\centering
\caption{Per-agent fused statistics at $\alpha=0.2$. $1{-}\widehat{m}_i$: coverage; $c_i$: classified (singleton) fraction; $\nu_i$: no-data fraction. Ring permutation is $1{-}2{-}3{-}5{-}4{-}1$, so cyclic neighbors share spatial coverage.}
\label{tab:topology_op}
\renewcommand{\arraystretch}{1.05}
\begin{tabular}{l|ccc|ccc}
\hline
 & \multicolumn{3}{c|}{Ring} & \multicolumn{3}{c}{Mesh} \\
agent & $1{-}\widehat{m}_i$ & $c_i$ & $\nu_i$ & $1{-}\widehat{m}_i$ & $c_i$ & $\nu_i$ \\
\hline
robot1 & 0.988 & 0.490 & 0.200 & 0.976 & 0.848 & 0.017 \\
robot2 & 0.988 & 0.594 & 0.133 & 0.976 & 0.848 & 0.017 \\
robot3 & 0.994 & 0.350 & 0.110 & 0.976 & 0.848 & 0.017 \\
robot4 & 0.984 & 0.587 & 0.205 & 0.976 & 0.848 & 0.017 \\
robot5 & 0.995 & 0.157 & 0.224 & 0.976 & 0.848 & 0.017 \\
\hline
mean   & 0.990 & 0.436 & 0.174 & 0.976 & 0.848 & 0.017 \\
\hline
\end{tabular}
\end{table}

\begin{figure*}[t]
  \centering
  \begin{subfigure}{0.19\textwidth}
    \includegraphics[width=\linewidth]{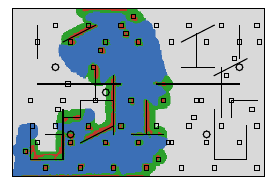}
    \caption{robot1}
  \end{subfigure}\hfill
  \begin{subfigure}{0.19\textwidth}
    \includegraphics[width=\linewidth]{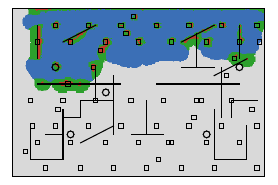}
    \caption{robot2}
  \end{subfigure}\hfill
  \begin{subfigure}{0.19\textwidth}
    \includegraphics[width=\linewidth]{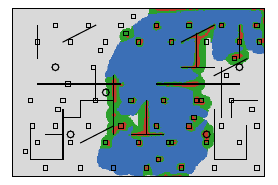}
    \caption{robot3}
  \end{subfigure}\hfill
  \begin{subfigure}{0.19\textwidth}
    \includegraphics[width=\linewidth]{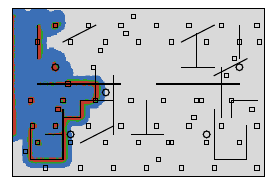}
    \caption{robot4}
  \end{subfigure}\hfill
  \begin{subfigure}{0.19\textwidth}
    \includegraphics[width=\linewidth]{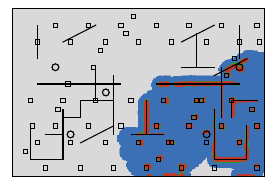}
    \caption{robot5}
  \end{subfigure}

  \medskip
  \begin{subfigure}{0.19\textwidth}
    \includegraphics[width=\linewidth]{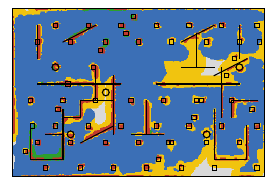}
    \caption{fused (mesh)}
  \end{subfigure}

  \medskip
  \includegraphics[width=0.7\textwidth]{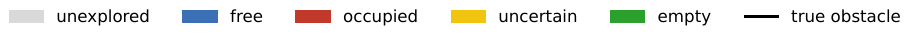}
  \caption{Per-agent local maps (a)--(e), each at the agent's own level $\alpha$, and the fused mesh map (f), at $\alpha=0.2$. Colours: unexplored (gray), free (blue), occupied (red), uncertain $\{-1,+1\}$ with data (yellow), empty $\varnothing$ (green); black lines mark the true obstacles (ground truth). Each agent's standalone conformal map classifies free/occupied within its observed region (gray elsewhere), but leaves empty cells where the single agent under-covers under the temporal shift. Fusion combines the local maps into a near-complete, valid map, with the residual uncertain label sitting mainly along obstacle boundaries.}
  \label{fig:maps}
    \vspace{-0.7cm}
\end{figure*}

Fig.~\ref{fig:maps} shows that each agent's standalone conformal map is gray (no data) outside its own route and, even within it, leaves empty (mis-covered) cells under the temporal shift, whereas the fused mesh map covers nearly the whole workspace and restores valid coverage, leaving $\{-1,+1\}$ mainly along obstacle boundaries. The ring's per-agent variation tracks spatial overlap between cyclic neighbors: robot~1, with overlapping neighbors, decides $49.0\%$ of queries, whereas robot~5, whose cyclic edge to robot~4 spans opposite ends of the workspace, decides only $15.7\%$ and returns the largest sets. The mesh removes this heterogeneity because $\mathcal{N}(i)\cup\{i\}$ is the full agent set for every $i$. Fig.~\ref{fig:coverage_alpha} plots empirical coverage against the nominal $1-\alpha$ over $\alpha\in\{0.05,\dots,0.40\}$. Both curves stay well above the ideal diagonal throughout, confirming Theorem~\ref{thm:fusion}, and remain near $0.98$ across the range. However, it should be noted that cells with a decision are $97$--$99\%$ accurate, and the residual abstention mass (uncertain or no-data sets) counts as covered, so realized coverage stays high even as $\alpha$ grows.
\begin{figure}[ht]
  \centering
  \includegraphics[width=\columnwidth]{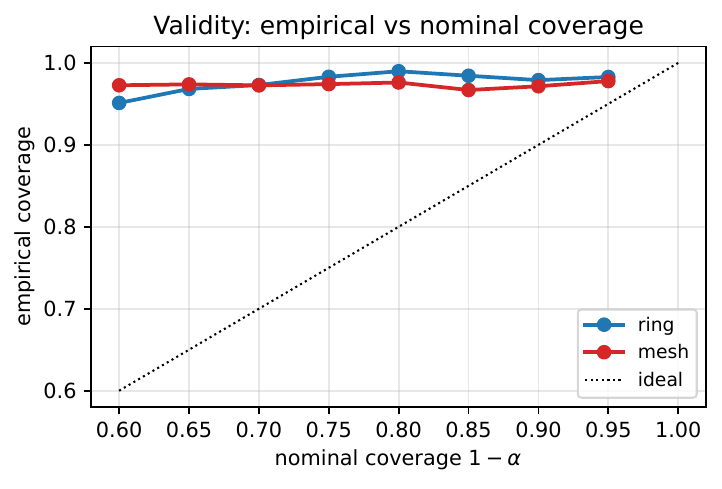}
  \caption{Empirical fused coverage versus nominal coverage $1-\alpha$ for ring and mesh. Both stay above the ideal line at every $\alpha$.}
  \label{fig:coverage_alpha}
    \vspace{-0.6cm}
\end{figure}
\begin{figure}[ht]
  \centering
  \includegraphics[width=\columnwidth]{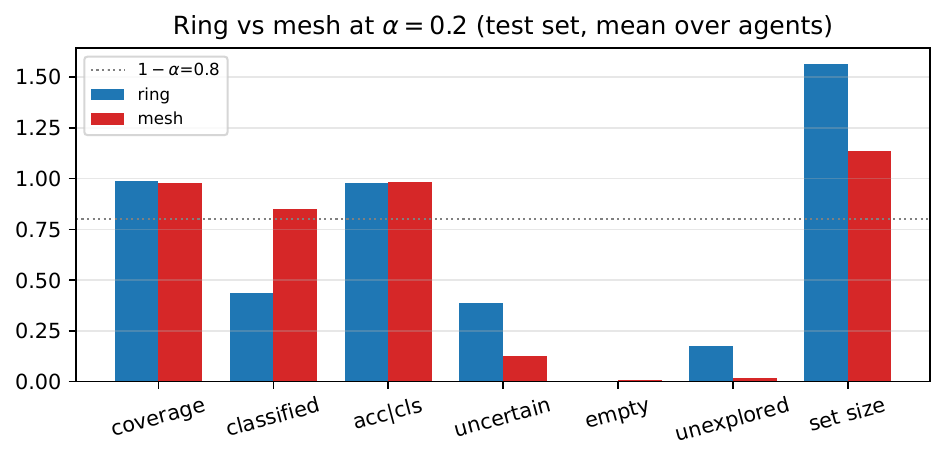}
  \caption{Ring vs.\ mesh at $\alpha=0.2$ (test set, mean over agents). The mesh determines more of the map (higher classification, lower uncertainty, and no-data fraction) at equal coverage and accuracy.}
  \label{fig:topology}
    \vspace{-0.5cm}
\end{figure}
\begin{figure}[ht]
  \centering
  \includegraphics[width=\columnwidth]{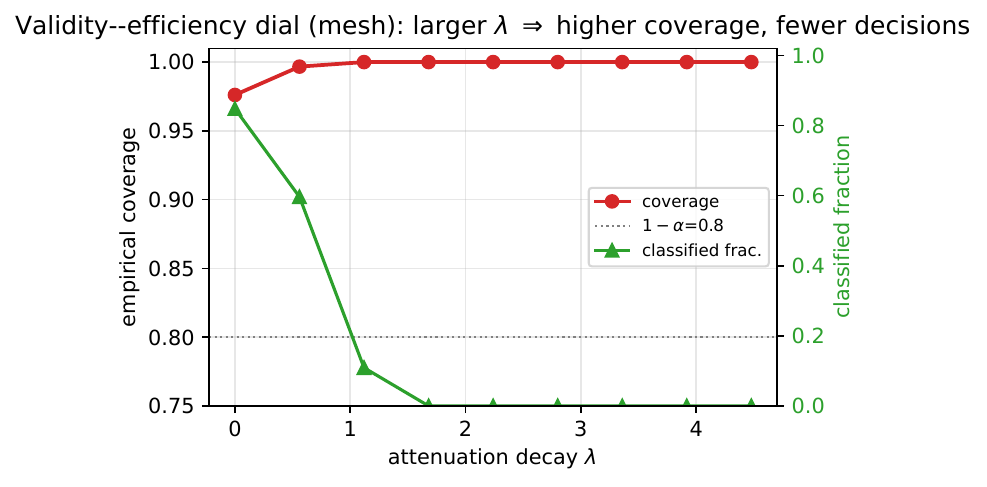}
  \caption{Increasing the attenuation decay $\lambda$ raises coverage toward $1$ while shrinking the classified fraction toward $0$ (sets grow to $\{-1,+1\}$).}
  \label{fig:recovery}
    \vspace{-0.6cm}
\end{figure}
Fig.~\ref{fig:recovery} sweeps the attenuation decay $\lambda$ at the mesh operating point. As $\lambda$ grows, every $g_k^j=e^{-\lambda\hat{\sigma}}$ shrinks, so the decided fraction falls from $0.85$ to $0$ while coverage rises from $0.98$ to $1$ --- $\lambda$ trades informativeness for conservativeness exactly as Remark~\ref{rem:dial} predicts. At $\lambda=0$ the method is already valid, so attenuation is a safety margin rather than a necessity in this benchmark.
The proposed fusion algorithm has three main limitations. First, the guarantee in Theorem~\ref{thm:fusion} is marginal: it ensures coverage on average over the test point and calibration data, but not conditional coverage at a fixed query location \(x^\star\). Such pointwise guarantees in continuous spaces generally require stronger assumptions on the data-generating process or score functions. Second, validity depends on local stationarity (Assumption~\ref{ass:stationarity}) and temporal mixing (Assumption~\ref{ass:stationarity}). In highly dynamic environments, or under erratic robot motion that increases the total-variation drift \(\delta_T\), one must tune the block length \(T\) or introduce larger safety margins. Third, although exchanging scalar e-values is far cheaper than sharing raw point clouds or model parameters, the communication cost still scales with the size of the shared query set \(|\mathcal{Q}_k|\), which may be limiting in large 3D environments.



\section{Conclusion}\label{sec:conclusion}


This paper presented a distributed fusion algorithm for multi-robot environmental mapping that combines Gaussian Process local estimators, conformal prediction, and e-value fusion to recover a user-specified finite-sample coverage level despite spatial and temporal correlations in each robot’s data. By decoupling local model training from distributed fusion, the method enables agents to construct set-valued occupancy maps with formal coverage guarantees, using a per-neighborhood budget and uncertainty attenuation to balance individual decisiveness with network-wide validity. Theoretical analysis shows that coverage holds independently of the communication graph topology, while simulations demonstrate that denser topologies improve the informative fraction of the map, helping bridge scalable continuous mapping with the reliability requirements of safety-critical autonomy.

\bibliographystyle{ieeetr}
\bibliography{biblio}

\end{document}